\documentclass[11pt]{article}
\usepackage{amsfonts}
\usepackage{graphics,amssymb,amsthm,amsmath,epsf}
\usepackage{graphicx}
\usepackage{subfigure}

\usepackage[round]{natbib}

\newtheorem{theorem}{Theorem}[section]
\newtheorem{prop}{Proposition}
\newtheorem{example}{Example}[section]
\newtheorem{corollary}{Corollary}[section]
\numberwithin{equation}{section}

\newtheorem{lemma}{Lemma}[section]

\allowdisplaybreaks

\begin{document}
\title{On Bivariate Exponentiated Extended \\ Weibull Family of Distributions}
\author{Rasool Roozegar\thanks{Corresponding: rroozegar@yazd.ac.ir } , Ali Akbar Jafari  \ \\
{\small Department of Statistics, Yazd University, Yazd,  Iran}\\}

\date{}
\maketitle

\begin{abstract}
In this paper, we introduce a new class of bivariate distributions called the bivariate exponentiated extended Weibull distributions. The model introduced here is of Marshall-Olkin type. This new class of bivariate distributions contains several bivariate lifetime models. Some mathematical properties of the new class of distributions are studied. We provide the joint and conditional density functions, the joint cumulative distribution function and the joint survival function. Special bivariate distributions are investigated in some detail. The maximum likelihood estimators are obtained using the EM algorithm. We illustrate the usefulness of the new class by means of application to two real data sets.
\end{abstract}

\section{Introduction}

The Weibull distribution has assumed a prominent position as statistical model for data from reliability, engineering and biological studies
\citep{mccool-12}.
%(McCool, 2012).
The Weibull distribution is a reasonable choice due to its negatively and positively skewed density shapes. However, this distribution is not a good model for describing phenomenon with non-monotone failure rates, which can be found on data from applications in reliability studies. Thus, extended forms of the Weibull model have been sought in many applied areas. As a solution for this issue, the inclusion of additional parameters to a well-defined distribution has been indicated as a good methodology for providing more flexible new classes of distributions.

The class of extended Weibull (EW) distributions pioneered by
\cite{gu-di-ra-97}
%Gurvich et al. (1997)
has achieved a prominent position in lifetime models. Its cumulative distribution function (cdf) is given by
\begin{equation}\label{eq.FG}
G\left(x;\lambda ,{\boldsymbol \xi }\right)=1-e^{-\lambda H(x;{\boldsymbol \xi })},\ \ \ \ \ x>0,\ \ \ \lambda >0,
\end{equation}
where $H(x;{\boldsymbol \xi })$ is a non-negative monotonically increasing function which depends on the parameter vector ${\boldsymbol \xi }$\textbf{. }The corresponding probability density function (pdf) is given by
\begin{equation}\label{eq.fG}
{\rm g}\left(x;\lambda ,{\boldsymbol \xi }\right)=\lambda  h\left(x;{\boldsymbol \xi }\right)e^{-\lambda H(x;{\boldsymbol \xi })},
\ \ \ \ \ x>0,\ \ \ \lambda >0,
\end{equation}
where $h\left(x;{\boldsymbol \xi }\right)$ is the derivative of $H(x;{\boldsymbol \xi })$. We emphasize that several distributions could be expressed in the form \eqref{eq.FG}. Table \ref{tab.model} summarizes several of these models. Further, we refer the reader to
\cite{na-ko-05}
%Nadarajah and Kotz (2005)
 and
 \cite{ph-la-07}.
 %Pham and Lai (2007).

In recent years, many authors worked on this class of distributions such as the beta extended Weibull family by
\cite{co-or-si-12beta},
 %Cordeiro, G. M., Ortega, E. M. M., & Silva, G. (2012)
the extended Weibull power series distributions by
\cite{si-bo-di-co-13},
% Silva et al. (2013),
the complementary extended Weibull power series class of distributions by
\cite{co-si-14},
 %Cordeiro and Silva (2014),
the Marshall-Olkin extended Weibull family of distributions by
\cite{sa-bo-ze-na-co-14}
 %Santos-Neto et al (2014)
and the exponentiated extended Weibull-power series class of distributions by
\cite{ta-ja-15Exponentiated}.
%  Tahmasebi and Jafari (2015).

\begin{table}[ht]
\caption{ Special cases of EW distributions and corresponding $H(x;{\boldsymbol \xi})$ function} \label{tab.model}
{\footnotesize
\begin{tabular}{|l|c|c|c|c|l|} \hline
Distribution & Support & $H\left(x;{\boldsymbol \xi }\right)$ & $\lambda $ & ${\boldsymbol \xi }$ & \multicolumn{1}{c|}{Reference} \\ \hline
Exponential & $x\geq 0$ & $x$ & $\lambda $ & $\varnothing $ & \cite{jo-ko-ba-95-1}%Johnson et al. (1995)
\\ \hline
Pareto & $x=k$ & ${\log  (x/k)}$ & $\lambda $ & $k$ & \cite{jo-ko-ba-95-1} %Johnson et al. (1994)
\\ \hline
Gompertz & $x\geq 0$ & $c^{-1}\left[{\exp  \left(cx\right)-1 }\right]$ & $\lambda $ & $c$ & \cite{gompertz-1825}%Gompertz (1825)
\\ \hline

Weibull & $x\geq 0$ & $x^{\gamma }$ & $\lambda $ & $\gamma $ & \cite{frechet-27}%Fréchet (1927)
\\ \hline

Fr\'{e}chet & $x\geq 0$ & $x^{-\gamma }$ & $\lambda $ & $\gamma $ & \cite{frechet-27}%Fréchet (1927)
\\ \hline

Lomax & $x\geq 0$ & $\log(1+x)$ & $\lambda $ & $\varnothing $ & \cite{lomax-54} \\ \hline

Weibull Kies & $0<\mu <x<\sigma $ & ${\left(x-\mu \right)}^{b_1}/{\left(\sigma -x\right)}^{b_2}$ & $\lambda $ & $\left(\mu ,\sigma ,b_1,b_2\right)$ &
 \cite{kies-58} %Kies (1958)
 \\ \hline

Log-logistic  &  $x\geq 0$ &  $\log(1+x^c)$ & $\lambda $ &  $c$  & \cite{fisk-61} \\ \hline

Linear failure rate & $x\geq 0$ & $ax+bx^2/2$ & 1 & $(a,b)$ & \cite{barlow-68}%Barlow (1968)
\\ \hline

Log-Weibull & $-\infty<x<\infty$ & ${\exp  [(x-\mu )/\sigma ] }$ & 1 & $\left(\mu ,\sigma \right)$ & \cite{white-69}%White (1969)
 \\ \hline
Exponential power & $x\geq 0$ & ${\exp  ({(cx)}^a-1) }$ & 1 & $(a,c)$ & \cite{sm-ba-75} %Smith and Bain (1975)
\\ \hline

Burr XII &  $x\geq 0$ & $\log(1+x^c)$ & $\lambda $ &  $c$  & \cite{rodriguez-77} \\ \hline

Rayleigh & $x\geq 0$ & $x^2$ & $\lambda $ & $\varnothing $ & \cite{rayleigh-80} %Rayleigh (1980)
 \\ \hline
Phani & $0<\mu <x<\sigma $ & ${[(x-\mu )/(\sigma -x)]}^b$ & $\lambda $ & $\left(\mu ,\sigma ,b\right)$ & \cite{phani-87} %Phani (1987)
\\ \hline
Additive Weibull & $x\geq 0$ & ${(x/{\beta }_1)}^{{\alpha }_1}+{(x/{\beta }_2)}^{{\alpha }_2} $ & 1 & $({\alpha }_1,{\alpha }_2,{\beta }_1,{\beta }_2)$ &
\cite{xi-ch-95} %Xie and Lai (1995)
\\ \hline
Chen & $x\geq 0$ & ${\exp  (x^b-1)}$ & $\lambda $ & $b$ & \cite{chen-00new} %Chen (2000)
\\ \hline
Pham & $x\geq 0$ & ${\left(a^x\right)}^{\beta }-1$ & 1 & $(a,\beta )$ & \cite{pham-02} %Pham (2002)
\\ \hline
Weibull extension & $x\geq 0$ & $c\left[\exp(cx)^b-1\right]$ & $\lambda $ & $(\gamma ,b,c)$ & \cite{xi-ta-go-02} %Xie et al. (2002)
\\ \hline
Modified Weibull & $x\geq 0$ & $x^{\gamma }{\exp  (cx)}$ & $\lambda $ & $(\gamma ,c)$ & \cite{la-xi-mu-03} %Lai et al. (2003)
\\ \hline
Traditional Wibull & $x\geq 0$ & $x^d{\exp  (cx^a-1)}$ & $\lambda $ & $(a,b,c)$ & \cite{na-ko-05} %Nadarajah and Kotz (2005)
\\ \hline
Generalized Weibull power & $x\geq 0$ & ${[1+{(x/a)}^b]}^c-1 $ & 1 & $(a,b,c)$ & \cite{ni-ha-06} %Nikulin and Haghighi (2006)
\\ \hline
Flexible Weibull extension & $x\geq 0$ & ${\exp  \left({\alpha }_1x-{\beta }_1/x\right) }$ & 1 & $({\alpha }_1,{\beta }_1)$ & \cite{be-la-zi-07} %Bebbington et al. (2007)
\\ \hline
Almalki Additive Weibull & $x\geq 0$ & $ax^{\theta }+bx^{\gamma }e^{cx}$ & 1 & $(a,b,c,\theta ,\gamma )$ & \cite{al-yu-13} % Almalki and Yuan (2013)
 \\ \hline
\end{tabular}
}
\end{table}

The aim of this paper is to introduce a new bivariate exponentiated extended Weibull (BEEW) family of distributions, whose marginals are exponentiated extended Weibull (EEW) distributions. It is obtained using a method similar to that used to obtain Marshall-Olkin bivariate exponential model
\citep{ma-ol-67}.
%Marshall and Olkin (1967).
The proposed BEEW class of distributions is constructed from three independent EEW distributions using a maximization process. Creating a bivariate distribution with given marginals using this technique is nothing new. The joint cdf can be expressed as a mixture of an absolutely continuous cdf and a singular cdf. The joint pdf of the BEEW distributions can take different shapes and the cdf can be expressed in a compact form. The joint cdf, the joint pdf and the joint survival function (sf) are in closed forms, which make it convenient to use in practice.
The new class of bivariate distributions contains as special models the bivariate generalized exponential
\citep{ku-gu-09-BGE},
%(Kundu and Gupta, 2009),
bivariate generalized linear failure rate
\citep{sa-ha-sm-ku-11},
%(Sarhan et al., 2011)
bivariate generalized Gompertz
\citep{elSh-ib-be-13},
%  (El-Sherpieny et al., 2013)
bivariate exponentiated generalized Weibull-Gompertz
\citep{elba-elda-mu-el-15},
% El-Bassioun et al. (2015)
bivariate exponentiated modified Weibull extension
\citep{elgo-elmo-15}
% El-Gohary,  El-Morshedy (2015)
distributions. This class defines at least 46 ($2\times23$) bivariate sub-models as special cases.

The usual maximum likelihood estimators can be obtained by solving non-linear equations in at least five unknowns directly, which is not a trivial issue. To avoid difficult computation we treat this problem as a missing value problem and use the EM algorithm, which can be implemented more conveniently than the direct maximization process. Another advantage of the EM algorithm is that it can be used to obtain the observed Fisher information matrix, which is helpful for constructing the asymptotic confidence intervals for the parameters. Alternatively, it is possible to obtain approximate maximum likelihood estimators by estimating the marginals first and then estimating the dependence parameter through a copula function, as suggested by
\citep[][Chapter 10]{joe-97},
%Joe (1997, Chapter 10),
which has the same rate of convergence as the maximum likelihood estimators. This is computationally less involved compared to the MLE calculations. This approach is not pursued here. Although in this paper we mainly discuss the BEEW, many of our results can be easily extended to the multivariate case.

The main reasons for introducing this new class of bivariate distributions are: (i) This class of distributions is an important model that can be used in a variety of problems in modeling bivariate lifetime data. (ii) It provides a reasonable parametric fit to skewed bivariate data that cannot be properly fitted by other distributions. (iii) The joint cdf and joint pdf should preferably have a closed form representation; at least numerical evaluation should be possible. (v) This class contains several special bivariate models because of the general class of Weibull distributions and the fact that the current generalization provides means of its bivariate continuous extension to still more complex situations; therefore it can be applied in modeling bivariate lifetime data.

The rest of the paper is organized as follows. We define the EEW and the BEEW class of distributions in Section \ref{sec.BEEW}. Different properties of this family are discussed in this section. The special cases of the BEEW model are considered in Section \ref{sec.sp}. The EM algorithm to compute the MLEs of the unknown parameters is provided in Section \ref{sec.mle}. The analysis of two real data sets are provided in Section \ref{sec.ex}. Finally, we conclude the paper in Section \ref{sec.con}.

\section{ The BEEW model}
\label{sec.BEEW}

In this section, we introduce the BEEW distributions using a method similar to that which was used by
\cite{ma-ol-67}
%Marshall and Olkin (1967)
to define the Marshall Olkin bivariate exponential (MOBE) distribution.

First, consider the univariate EEW class of distributions with cdf given by
\begin{equation}\label{eq.FEEW}
F_{{\rm EEW}}\left(x;\alpha ,\lambda ,{\boldsymbol \xi }\right)={\left(1-e^{-\lambda H(x;{\boldsymbol \xi })}\right)}^{\alpha },\ \ \ \ \ x>0,\ \ \alpha >0,\ \ \lambda >0.
\end{equation}
The corresponding pdf is
\begin{equation}\label{eq.fEEW}
f_{{\rm E}{\rm EW}}\left(x;\alpha ,\lambda ,{\boldsymbol \xi }\right)=\alpha \lambda \ h\left(x;{\boldsymbol \xi }\right)e^{-\lambda H(x;{\boldsymbol \xi })}{\left(1-e^{-\lambda H(x;{\boldsymbol \xi })}\right)}^{\alpha -1}.
\end{equation}

From now on a EEW class of distributions with the shape parameter $\alpha$, the scale parameter $\lambda$ and parameter vector ${\boldsymbol \xi }$ will be denoted by ${\rm EEW}(\alpha ,\lambda ,{\boldsymbol \xi })$. Note that many well-known models could be expressed in the form
\eqref{eq.FEEW}, such as exponentiated Weibull
\citep{mu-sr-93},
%(Mudholkar and Srivastava, 1993),
 generalized exponential
 \citep{gu-ku-99},
% (Gupta and Kundu, 1999),
Weibull extension
\citep{chen-00new},
%(Chen, 2000),
generalized Rayleigh
\citep{su-pa-01,ku-ra-05},
%(Surles and Padgett, 2001; Kundu and Raqab, 2005),
modified Weibull extension
\citep{xi-ta-go-02},
%(Xie et al., 2002),
generalized modified Weibull
\citep{ca-or-co-08}
%(Carrasco et al., 2008),
generalized linear failure rate
\citep{sa-ku-09},
%(Sarhan and Kundu, 2009),
generalized Gompertz
\citep{el-al-al-13},
%(El-Gohary et al., 2013),
and exponentiated modified Weibull extension
\citep{sa-ap-13}
%(Sarhan and Apaloo, 2013)
 distributions.

When $\alpha $ is a positive integer, the EEW model can be interpreted as the lifetime distribution of a parallel system consisting of $\alpha $ independent and identical units whose lifetime follows the EEW distributions.

From now on unless otherwise mentioned, it is assumed that ${\alpha }_1>0$; ${\alpha }_2>0$; ${\alpha }_3>0$ and $\lambda >0$. Suppose $U_1\sim {\rm EEW}({\alpha }_1,\lambda ,{\boldsymbol \xi })$, $U_2\sim {\rm EEW}({\alpha }_2,\lambda ,{\boldsymbol \xi })$ and $U_3\sim {\rm EEW}({\alpha }_3,\lambda ,{\boldsymbol \xi })$ and they are mutually independent. Here ``$\sim$'' means follows or has the distribution. Now define $X_1= \max\{U_1,U_3\}$ and
$X_2= \max\{U_2,U_3\}$. Then, we say that the bivariate vector $(X_1,X_2)$  has a bivariate exponentiated extended Weibull distribution with the shape parameters $\alpha_1$, ${\alpha }_2$ and ${\alpha }_3$, the scale parameter $\lambda $ and parameter vector ${\boldsymbol \xi }$. We will denote it by
${\rm BEEW}({\alpha }_1,{\alpha }_2,{\alpha }_3,\lambda ,{\boldsymbol \xi })$. Before providing the joint cdf or pdf, we first mention how it may occur in practice.

According to
\cite{ku-gu-09-BGE},
%Kundu and Gupta (2009),
suppose a system has two components and it is assumed that each component has been maintained independently and also there is an overall maintenance. Due to component maintenance, suppose the lifetime of the individual component is increased by $U_i$ amount and because of the overall maintenance, the lifetime of each component is increased by $U_3$ amount. Therefore, the increased lifetimes of the two component are $X_1= \max\{U_1,U_3\}$ and
$X_2= \max\{U_2,U_3\}$, respectively.

We now study the joint cdf of the bivariate random vector $(X_1,X_2)$ in the following theorem.

\begin{theorem} If $(X_1,X_2{\boldsymbol )\sim }{\rm BEEW}({\alpha }_1,{\alpha }_2,{\alpha }_3,\lambda ,{\boldsymbol \xi })$, then the joint cdf of
$(X_1,X_2)$ for $x_1>0$, $x_2>0$, is
\begin{equation}\label{eq.FBEEW}
F_{\rm BEEW}\left(x_1, x_2\right)={\left(1-e^{-\lambda H(x_1;{\boldsymbol \xi })}\right)}^{{\alpha }_1}{\left(1-e^{-\lambda H(x_2;{\boldsymbol \xi})}\right)}^{{\alpha }_2}{(1-e^{-\lambda H(z;{\boldsymbol \xi })})}^{{\alpha }_3},
\end{equation}
where $z= \min\{x_1,x_2\}$.
\end{theorem}

\begin{proof}
Since the joint cdf of the random variables $X_1$ and $X_2$ is defined as
\begin{eqnarray*}
F_{{\rm BEEW}}\left(x_1,\ x_2\right)&=&P\left(X_1\le x_1,\ X_2\le x_2\right)\\
&=&P\left({\max  \left\{U_1,U_3\right\}\ }\le x_1,{\max  \left\{U_2,U_3\right\}\ }\le x_2\right)\\
&=&P(U_1\le x_1,U_2\le x_2,U_3\le {\rm min}(x_1,x_2)).
\end{eqnarray*}
As the random variables $U_i$, $(i=1,2,3)$ are mutually independent, we directly obtain
\begin{equation}\label{eq.FF}
F_{\rm BEEW}(x_1,\ x_2;\alpha_1,\alpha_2,\alpha_3,\lambda,{\boldsymbol \xi })
=F_{\rm EEW}(x_1;{\alpha }_1,\lambda ,{\boldsymbol \xi })
 F_{\rm EEW}(x_2;{\alpha }_2,\lambda ,{\boldsymbol \xi })
 F_{\rm EEW}(z  ;{\alpha }_3,\lambda ,{\boldsymbol \xi }).
\end{equation}
Substituting from \ref{eq.FEEW} into \ref{eq.FF}, we obtain \ref{eq.FBEEW}, which completes the proof of the theorem.
\end{proof}

\begin{corollary}
The joint cdf the ${\rm BEEW}({\alpha }_1,{\alpha }_2,{\alpha }_3,\lambda ,{\boldsymbol \xi })$ can also written as
\begin{eqnarray}\label{eq.FBEEW2}
F_{{\rm BEEW}}\left(x_1,x_2\right)&=&\left\{ \begin{array}{ll}
{\left(1-e^{-\lambda H\left(x_1;{\boldsymbol \xi }\right)}\right)}^{{\alpha }_1+{\alpha }_3}{\left(1-e^{-\lambda H\left(x_2;{\boldsymbol \xi }\right)}\right)}^{{\alpha }_2} & {\rm if}\ \ \ x_1\le x_2 \\
{\left(1-e^{-\lambda H\left(x_1;{\boldsymbol \xi }\right)}\right)}^{{\alpha }_1}{\left(1-e^{-\lambda H\left(x_2;{\boldsymbol \xi }\right)}\right)}^{{\alpha }_2+{\alpha }_3} & {\rm if}\ \ \ x_1>x_2 \end{array}
\right.\\
&=&\left\{ \begin{array}{ll}
F_{{\rm EEW}}\left(x_1;{\alpha }_1+{\alpha }_3,\lambda ,{\boldsymbol \xi }\right)F_{{\rm EEW}}\left(x_2;{\alpha }_2,\lambda ,{\boldsymbol \xi }\right) &
{\rm if}\ \ \  x_1<x_2
 \\
F_{{\rm EEW}}\left(x_1;{\alpha }_1,\lambda ,{\boldsymbol \xi }\right)F_{{\rm EEW}}\left(x_2;{\alpha }_2+{\alpha }_3,\lambda ,{\boldsymbol \xi }\right) &
{\rm if}\ \ \  x_2<x_1
 \\
F_{{\rm EEW}}\left(x;{{\alpha }_1+\alpha }_2+{\alpha }_3,\lambda ,{\boldsymbol \xi }\right) &
{\rm  if}\ \ \ x_1=x_2=x.
\end{array}
\right.\nonumber
\end{eqnarray}
\end{corollary}

The following theorem gives the joint pdf of the random variables $X_1$ and $X_2$ which is the joint pdf of ${\rm BEEW}({\alpha }_1,{\alpha }_2,{\alpha }_3,\lambda ,{\boldsymbol \xi })$.

\begin{theorem} \label{thm.fX1X2}
If $(X_1,X_2)\sim {\rm BEEW}({\alpha }_1,\alpha_2,{\alpha }_3,\lambda ,{\boldsymbol \xi })$ then the joint pdf of
$(X_1,X_2)$ for $x_1>0$, $x_2>0$, is
\begin{equation}\label{eq.fBEEW}
f_{{\rm BEEW}}\left(x_1,x_2\right)=
\left\{ \begin{array}{ll}
f_{{\rm 1}}\left(x_1,x_2\right) & {\rm if}\ \ \ \ 0<\ x_1<x_2\\
f_{{\rm 2}}\left(x_1,x_2\right) & {\rm if}\ \ \ \ 0<\ x_2<x_1\\
f_0\left(x\right)               & {\rm if}\ \ \ \ 0<x_1=x_2=x, \end{array}
\right.
\end{equation}
where
\begin{eqnarray}
f_{{\rm 1}}\left(x_1,x_2\right)&=&f_{{\rm EEW}}\left(x_1;{\alpha }_1+{\alpha }_3,\lambda ,{\boldsymbol \xi }\right)f_{{\rm EEW}}\left(x_2;{\alpha }_2,\lambda ,{\boldsymbol \xi }\right)\nonumber\\
&=&
\left({\alpha }_1+{\alpha }_3\right){\alpha }_2{\lambda }^2\ h\left(x_1;{\boldsymbol \xi }\right)h\left(x_2;{\boldsymbol \xi }\right){\left(1-e^{-\lambda H(x_1;{\boldsymbol \xi })}\right)}^{{\alpha }_1+{\alpha }_3-1}\nonumber\\
&&\times
 {\left(1-e^{-\lambda H(x_2;{\boldsymbol \xi })}\right)}^{{\alpha }_2-1} e^{-\lambda H(x_1;{\boldsymbol \xi })-\lambda H(x_2;{\boldsymbol \xi })}\label{eq.f1}\\
f_{{\rm 2}}\left(x_1,x_2\right)&=&f_{{\rm EEW}}\left(x_1;{\alpha }_1,\lambda ,{\boldsymbol \xi }\right)f_{{\rm EEW}}\left(x_2;{\alpha }_2+{\alpha }_3,\lambda ,{\boldsymbol \xi }\right)\nonumber\\
&=&\left({\alpha }_2+{\alpha }_3\right){\alpha }_1{\lambda }^2\ h\left(x_1;{\boldsymbol \xi }\right)h\left(x_2;{\boldsymbol \xi }\right){\left(1-e^{-\lambda H(x_1;{\boldsymbol \xi })}\right)}^{{\alpha }_1-1}\nonumber\\
&&\times
{\left(1-e^{-\lambda H(x_2;{\boldsymbol \xi })}\right)}^{{\alpha }_2+{\alpha }_3-1} e^{-\lambda H(x_1;{\boldsymbol \xi })-\lambda H(x_2;{\boldsymbol \xi })}\label{eq.f2}\\
f_0\left(x\right)
&=&\frac{{\alpha }_3}{{\alpha }_1+{\alpha }_2+{\alpha }_3}f_{{\rm EEW}}\left(x;{\alpha }_1{+\alpha }_2+{\alpha }_3,\lambda ,{\boldsymbol \xi }\right)\nonumber\\
&=&{\alpha }_3\lambda \ h\left(x;{\boldsymbol \xi }\right){\left(1-e^{-\lambda H(x;{\boldsymbol \xi })}\right)}^{{\alpha }_1+{\alpha }_2+{\alpha }_3-1} e^{-\lambda H(x;{\boldsymbol \xi })}.\label{eq.f0}
\end{eqnarray}
\end{theorem}

\begin{proof}
First assume that $x_1<x_2$. Then, the expression for $f_{{\rm 1}}\left(x_1,x_2\right)$ can be obtained simply by differentiating the joint cdf  $F_{{\rm BEEW}}\left(x_1, x_2\right)$ given in \eqref{eq.FBEEW2} with respect to $x_1$ and $x_2$. Similarly, we find the expression of $f_{{\rm 2}}\left(x_1,x_2\right)$ when $x_2<x_1$. But $f_0\left(x\right)$ cannot be derived in the same way. Using the facts that
\begin{eqnarray*}
&&\int^{\infty }_0{\int^{x_2}_0{f_{{\rm 1}}\left(x_1,x_2\right)\ dx_1dx_2}}+\int^{\infty }_0{\int^{x_1}_0{f_{{\rm 2}}\left(x_1,x_2\right)\ dx_2dx_1}}+\int^{\infty }_0{f_0\left(x\right)\ dx}=1,\\
&&\int^{\infty }_0{\int^{x_2}_0{f_{{\rm 1}}\left(x_1,x_2\right)\ dx_1dx_2}}={\alpha }_2\int^{\infty }_0{\lambda \ h\left(x;{\boldsymbol \xi }\right){\left(1-e^{-\lambda H(x;{\boldsymbol \xi })}\right)}^{{\alpha }_1+{\alpha }_2+{\alpha }_3-1}\ \ e^{-\lambda H(x;{\boldsymbol \xi })}\ dx},
\end{eqnarray*}
and
\[\int^{\infty }_0{\int^{x_1}_0{f_{{\rm 2}}\left(x_1,x_2\right)\ dx_2dx_1}}={\alpha }_1\int^{\infty }_0{\lambda \ h\left(x;{\boldsymbol \xi }\right){\left(1-e^{-\lambda H(x;{\boldsymbol \xi })}\right)}^{{\alpha }_1+{\alpha }_2+{\alpha }_3-1}\ \ e^{-\lambda H(x;{\boldsymbol \xi })}\ dx}.\]
Note that
\[\int^{\infty }_0{f_0\left(x\right)\ dx}={\alpha }_3\int^{\infty }_0{\lambda \ h\left(x;{\boldsymbol \xi }\right){\left(1-e^{-\lambda H(x;{\boldsymbol \xi })}\right)}^{{\alpha }_1+{\alpha }_2+{\alpha }_3-1}\ \ e^{-\lambda H(x;{\boldsymbol \xi })}\ dx}=\frac{{\alpha }_3}{{\alpha }_1+{\alpha }_2+{\alpha }_3}.\]
Thus, the result follows.
\end{proof}

The following theorem gives the marginal pdf's of $X_1$ and $X_2$.

\begin{theorem} \label{thm.fXi}
The marginal distributions of $X_1$ and $X_2$ are ${\rm EEW}(\alpha_1+\alpha_3,\lambda ,\boldsymbol \xi)$ and
$EEW(\alpha_2$ $+\alpha_3,\lambda,\boldsymbol \xi)$, respectively.
\end{theorem}

\begin{proof}
The marginal cdf for $X_i$ is
\[F_{X_i}\left(x_i\right)=P\left(X_i\le x_i\right)=P\left({\max  \left\{U_i,U_3\right\}\ }\le x_i\right)=P(U_i\le x_i,U_3\le x_i).\]
Since the random variables $U_i$, $(i=1,2)$ are mutually independent, we obtain
\begin{eqnarray}\label{eq.FXi}
F_{X_i}\left(x_i\right)&=&P\left(U_i\le x_i\right)P(U_3\le x_i)\nonumber\\
&=&F_{{\rm EEW}}\left(x_i;{\alpha }_i,\lambda ,{\boldsymbol \xi }\right)F_{{\rm EEW}}\left(x_i;{\alpha }_3,\lambda ,{\boldsymbol \xi }\right)\nonumber\\
&=&
F_{{\rm EEW}}\left(x_i;{\alpha }_i+{\alpha }_3,\lambda ,{\boldsymbol \xi }\right).
\end{eqnarray}
From \ref{eq.FXi}, we can derive the pdf of $X_i$ by differentiation.
\end{proof}

The BEEW model has both an absolute continuous part and a singular part, similar to Marshall and Olkin's bivariate exponential model. The joint cdf of $X_1$ and $X_2$ has a singular part along the line $x_1=x_2$, with weight $\frac{{\alpha }_3}{{\alpha }_1+{\alpha }_2+{\alpha }_3}$, and has an absolutely continuous part on $0<x_1\ne x_2<\infty $ with weight $\frac{{\alpha }_1+{\alpha }_2}{{\alpha }_1+\alpha_2+{\alpha }_3}$.

Interestingly, the BEEW model can be obtained by using the Marshall\_Olkin (MO) copula with the marginals as the EEW distributions. To every bivariate cdf $F_{X_1,X_2}$with continuous marginals $F_{X_1}$ and $F_{X_2}$there corresponds a unique bivariate cdf with uniform margins $C:{\left[0,1\right]}^2\to [0,1]$ called a copula, such that $F_{X_1,X_2}\left(x_1,x_2\right)=C\{F_{X_1}\left(x_1\right),F_{X_2}\left(x_2\right)\}$ holds for all $\left(x_1,x_2\right)\in {{\mathbb R}}^2$
\citep{nelson-99}.
%(Nelsen, 1999).
The MO copula is
\[C_{{\theta }_1,{\theta }_2}\left(u_1,u_2\right)=u^{1-{\theta }_1}_1u^{1-{\theta }_2}_2\min  \left\{u^{{\theta }_1}_1,u^{{\theta }_2}_2\right\} ,\]
for $0<{\theta }_1<1$ and $0<{\theta }_2<1$. Using $u_i=F_{X_i}\left(x_i\right)$ where $X_i$ is ${\rm EEW}({\alpha }_i+{\alpha }_3,\lambda ,\xi )$ and ${\theta }_i=\frac{{\alpha }_3}{{\alpha }_i+{\alpha }_3}$, $i=1,2,3,$ gives the same joint cdf $F_{X_1,X_2}$ as \eqref{eq.FBEEW2}.

The following result will provide explicitly the absolute continuous part and the singular part of the BEEW cdf.

\begin{theorem}
If $ (X_1,X_2 )\sim  {\rm BEEW}({\alpha }_1,{\alpha }_2,{\alpha }_3,\lambda ,{\boldsymbol \xi })$,
then
\[F_{X_1,X_2}\left(x_1,x_2\right)=\frac{{\alpha }_1+{\alpha }_2}{{\alpha }_1+{\alpha }_2+{\alpha }_3}F_a\left(x_1,x_2\right)+\frac{{\alpha }_3}{{\alpha }_1+{\alpha }_2+{\alpha }_3}F_s\left(x_1,x_2\right),\]
where for $x={\min  \left\{x_1,x_2\right\}}$,
\[F_s\left(x_1,x_2\right)={\left(1-e^{-\lambda H\left(x;{\boldsymbol \xi }\right)}\right)}^{{\alpha }_1+{\alpha }_2+{\alpha }_3},\]
and
\begin{eqnarray*}
F_a\left(x_1,x_2\right)&=&\frac{{\alpha }_1+{\alpha}_2+{\alpha }_3}{{\alpha }_1+{\alpha }_2}{\left(1-e^{-\lambda H\left(x_1;{\boldsymbol \xi }\right)}\right)}^{{\alpha }_1}{\left(1-e^{-\lambda H\left(x_2;{\boldsymbol \xi }\right)}\right)}^{{\alpha }_2}{\left(1-e^{-\lambda H\left(x;{\boldsymbol \xi }\right)}\right)}^{{\alpha }_3}\\
&&
-\frac{{\alpha }_3}{{\alpha }_1+{\alpha }_2}{\left(1-e^{-\lambda H\left(x;{\boldsymbol \xi }\right)}\right)}^{{\alpha }_1+{\alpha }_2+{\alpha }_3},
\end{eqnarray*}
here $F_s\left(.,.\right)$ and $F_a\left(.,.\right)$ are the singular and the absolute continuous parts, respectively.
\end{theorem}

\begin{proof}
To find $F_a\left(x_1,x_2\right)$ from $F_{X_1,X_2}\left(x_1,x_2\right)=aF_a\left(x_1,x_2\right)+\left(1-a\right)F_s\left(x_1,x_2\right),\ 0\le a\le 1,$ we compute
\[\frac{{\partial }^2F_{X_1,X_2}\left(x_1,x_2\right)}{\partial x_1\ \partial x_2}=af_{a}\left(x_1,x_2\right)
=\left\{ \begin{array}{ll}
f_1\left(x_1,x_2\right)&  {\rm if}\ \ \ x_1<x_2 \\
f_2\left(x_1,x_2\right)&  {\rm if}\ \ \ x_1>x_2, \end{array}
\right.\]
from which $a$ may be obtained as
\[a=\int^{\infty }_0{\int^{x_2}_0{f_{{\rm 1}}\left(x_1,x_2\right)\ dx_1dx_2}}+\int^{\infty }_0{\int^{x_1}_0{f_{{\rm 2}}\left(x_1,x_2\right)\ dx_2dx_1}}=\frac{{\alpha }_1+{\alpha }_2}{{\alpha }_1+{\alpha }_2+{\alpha }_3},\]
and
\[F_a\left(x_1,x_2\right)=\int^{x_1}_0{\int^{x_2}_0{f_{{\rm a}}\left(s,t\right)\ ds\ dt}.}\]
Once $a$ and $F_a\left(.,.\right)$ are determined, $F_s\left(.,.\right)$ can be obtained by subtraction.
\end{proof}

\begin{corollary}
The joint pdf of $X_1$ and $X_2$ can be written as follows for $x={\min  \left\{x_1,x_2\right\}};$
\[f_{X_1,X_2}\left(x_1,x_2\right)=\frac{{\alpha }_1+{\alpha }_2}{{\alpha }_1+{\alpha }_2+{\alpha }_3}f_a\left(x_1,x_2\right)+\frac{{\alpha }_3}{{\alpha }_1+{\alpha }_2+{\alpha }_3}f_s\left(x\right),\]
where
\[f_a\left(x_1,x_2\right)=\frac{{\alpha }_1+{\alpha }_2+{\alpha }_3}{{\alpha }_1+{\alpha }_2}\times
 \left\{ \begin{array}{ll}
f_{{\rm EEW}}\left(x_1;{\alpha }_1+{\alpha }_3,\lambda ,{\boldsymbol \xi }\right)f_{{\rm EEW}}\left(x_2;{\alpha }_2,\lambda ,{\boldsymbol \xi }\right)
& {\rm if}\ \ \ x_1<x_2, \\
f_{{\rm EEW}}\left(x_1;{\alpha }_1,\lambda ,{\boldsymbol \xi }\right)f_{{\rm EEW}}\left(x_2;{\alpha }_2+{\alpha }_3,\lambda ,{\boldsymbol \xi }\right)
& {\rm if} \ \ \ x_1>x_2, \end{array}
\right.\]
and
\[f_s\left(x\right)=f_{{\rm EEW}}\left(x;{{\alpha }_1+\alpha }_2+{\alpha }_3,\lambda ,{\boldsymbol \xi }\right).\]
Clearly, here $f_a\left(x_1,x_2\right)$ and $f_s\left(x\right)$ are the absolute continuous part and singular part, respectively.
\end{corollary}

Having obtained the marginal pdf of $X_1$ and $X_2$, we can now derive the pdf's as presented in the following theorem.

\begin{theorem}\label{fXicondXj}
The conditional pdf of $X_i$ given $X_j=x_j$, denoted by $f_{X_i|X_j}\left(x_i|x_j\right)$, $i\ne j=1,2,$ is given by
\begin{equation}
f_{X_i|X_j}\left(x_i|x_j\right)=\left\{ \begin{array}{ll}
f^{\left(1\right)}_{X_i|X_j}\left(x_i|x_j\right) & {\rm if}\ \ \ 0<x_i<x_j \\
f^{\left(2\right)}_{X_i|X_j}\left(x_i|x_j\right) & {\rm if}\ \ \ 0<x_j<x_i \\
f^{\left(3\right)}_{X_i|X_j}\left(x_i|x_j\right) & {\rm if}\ \ \ x_i=x_j>0, \end{array}
\right.
\end{equation}
where
\begin{eqnarray*}
&&f^{\left(1\right)}_{X_i|X_j}\left(x_i|x_j\right)=\frac{\left({\alpha }_i+{\alpha }_3\right){\alpha }_j\lambda \ h\left(x_i;{\boldsymbol \xi }\right){\left(1-e^{-\lambda H(x_i;{\boldsymbol \xi })}\right)}^{{\alpha }_i+{\alpha }_3-1} e^{-\lambda H(x_i;{\boldsymbol \xi })}}{({\alpha }_{?2}+{\alpha }_3)\ {\left(1-e^{-\lambda H(x_j;{\boldsymbol \xi })}\right)}^{{\alpha }_3} },\\
&&
f^{\left(2\right)}_{X_i|X_j}\left(x_i|x_j\right)={\alpha }_i\lambda \ h\left(x_i;{\boldsymbol \xi }\right){\left(1-e^{-\lambda H(x_i;{\boldsymbol \xi })}\right)}^{{\alpha }_i-1} e^{-\lambda H(x_i;{\boldsymbol \xi })}\\
&&f^{\left(3\right)}_{X_i|X_j}\left(x_i|x_j\right)=\frac{{\alpha }_3 }{{\alpha }_j+{\alpha }_3}{\left(1-e^{-\lambda H(x_i;{\boldsymbol \xi })}\right)}^{{\alpha }_i}.
\end{eqnarray*}
\end{theorem}

\begin{proof}
The proof follows readily upon substituting the joint pdf of $(X_1,X_2)$ given in Theorem \ref{thm.fX1X2} and the marginal pdf of $X_j$, given in Theorem \ref{thm.fXi}, using the following relation
\begin{equation}\label{eq.fcond}
f_{X_i|X_j}\left(x_i|x_j\right)=\frac{f_{X_i,X_j}(x_i,x_j)}{f_{X_i}(x_i)},\ \ i=1,2.
\end{equation}
\end{proof}

\begin{prop} Since the joint sf and the joint cdf have the following relation
\begin{equation}\label{eq.sX1X2}
S_{X_1,X_2}\left(x_1,x_2\right)=1-F_{X_1}\left(x_1\right)-F_{X_2}\left(x_2\right)+F_{X_1,X_2}\left(x_1,x_2\right),
\end{equation}
therefore, the joint sf of $X_1$ and $X_2$ also can be expressed in a compact form.
\end{prop}

\begin{prop}
\cite{basu-71}
% Basu (1971)
  defined the bivariate failure rate function $h_{X_1,X_2}\left(x_1,x_2\right)$ for the random vector $ (X_1,X_2 )$ as the following relation
\begin{equation}\label{eq.hX1X2}
h_{X_1,X_2}\left(x_1,x_2\right)=\frac{f_{X_1,X_2}\left(x_1,x_2\right)}{S_{X_1,X_2}\left(x_1,x_2\right)}.
\end{equation}
We can obtained the bivariate failure rate function $h_{X_1,X_2}(x_1,x_2)$ for the random vector $(X_1,$ $X_2)$
by substituting from \eqref{eq.fBEEW} and \eqref{eq.sX1X2} in \eqref{eq.hX1X2}.
\end{prop}

\begin{lemma} The cdf of $Y=\max  \{X_1,X_2\}$ is given as
\begin{equation}\label{eq.Fmax}
F_Y\left(y\right)={\left(1-e^{-\lambda H\left(y;{\boldsymbol \xi }\right)}\right)}^{{\alpha }_1+{\alpha }_2+{\alpha }_3}.
\end{equation}
\end{lemma}

\begin{proof}
Since
\begin{eqnarray*}
F_Y\left(y\right)&=&P\left(Y\le y\right)=P\left({\max  \{X_1,X_2\} }\le y\right)\\
&=&P\left(X_1\le y,X_2\le y\right)=P\left({\max  \{U_1,U_3\}\ }\le y,{\max  \{U_2,U_3\}\ }\le y\right)\\
&=&P\left(U_1\le y,U_2\le y,U_3\le y\right),
\end{eqnarray*}
where the random variables $U_i\ (i=1,2,3)$ are mutually independent, we directly obtain the result.
\end{proof}

%\textbf{Comment 2.3. }From lemma 2.1, we can say that, if $X_1$ and $X_2$ are independent EEW random variables then ${\max  \{X_1,X_2\}\ }$ is also EEW random variable.

\begin{lemma}
 The cdf of $T={\min  \{X_1,X_2\}}$ is given as
\[F_T\left(t\right)={\left(1-e^{-\lambda H\left(t;{\boldsymbol \xi }\right)}\right)}^{{\alpha }_1+{\alpha }_3}+{\left(1-e^{-\lambda H\left(t;{\boldsymbol \xi }\right)}\right)}^{{\alpha }_2+{\alpha }_3}-{\left(1-e^{-\lambda H\left(t;{\boldsymbol \xi }\right)}\right)}^{{\alpha }_1+{\alpha }_2+{\alpha }_3}.\]
\end{lemma}
\begin{proof}
It is easy to prove that by using Equations \eqref{eq.sX1X2} and \eqref{eq.Fmax}.
\end{proof}

\section{Special cases}
\label{sec.sp}
In this Section, we consider some special cases of the BEEW distributions.

\subsection{Bivariate generalized exponential distribution}

If $H\left(x;{\boldsymbol \xi }\right)=x,$ then the joint cdf \eqref{eq.FBEEW2} becomes
\[F\left(x_1,\ x_2\right)=\left\{ \begin{array}{ll}
{\left(1-e^{-\lambda x_1}\right)}^{{\alpha }_1+{\alpha }_3}{\left(1-e^{-\lambda x_2}\right)}^{{\alpha }_2} & {\rm if}\ \ \  x_1\le x_2 \\
{\left(1-e^{-\lambda x_1}\right)}^{{\alpha }_1}{\left(1-e^{-\lambda x_2}\right)}^{{\alpha }_2+{\alpha }_3} & {\rm if}\ \ \  x_1>x_2, \end{array}
\right.\]
which is the joint cdf of bivariate generalized exponential (BGE) distribution introduced by
\cite{ku-gu-09-BGE}.
%Kundu and Gupta (2009).
By Theorem
\ref{eq.FXi},
the marginal distributions of $X_1$ and $X_2$ are ${\rm GE}({\alpha }_1+{\alpha }_3,\lambda )$ and ${\rm GE}\left({\alpha }_2+{\alpha }_3,\lambda \right),$ respectively.

\subsection{ Bivariate generalized linear failure rate distribution}

If $H\left(x;{\boldsymbol \xi }\right)={\beta }x+\frac{\gamma }{2 }x^2$ and $\lambda=1$, then the joint cdf \eqref{eq.FBEEW2} becomes
\[F\left(x_1, x_2\right)=\left\{ \begin{array}{cc}
{\left(1-e^{-\beta x_1-\frac{\gamma }{2}x^2_1}\right)}^{{\alpha }_1+{\alpha }_3}{\left(1-e^{-\beta x_2-\frac{\gamma }{2}x^2_2}\right)}^{{\alpha }_2} & {\rm if}\ \ \  x_1\le x_2 \\
{\left(1-e^{-\beta x_1-\frac{\gamma }{2}x^2_1}\right)}^{{\alpha }_1}{\left(1-e^{-\beta x_2-\frac{\gamma }{2}x^2_2}\right)}^{{\alpha }_2+{\alpha }_3} & {\rm if}\ \ \  x_1>x_2, \end{array}
\right.\]
which is the joint cdf of bivariate generalized linear failure rate (BGLFR) distribution introduced by
\cite{sa-ha-sm-ku-11}.
%Sarhan et al. (2011).
 By Theorem \ref{eq.FXi}, the marginal distributions of $X_1$ and $X_2$ are ${\rm GLFR}({\alpha }_1+{\alpha }_3,\beta ,\gamma )$ and ${\rm GLFR}\left({\alpha }_2+{\alpha }_3,\beta ,\gamma \right),$ respectively.

\subsection{Bivariate exponentiated Weibull distribution}

%If ${\rm H}\left({\rm x};\xi \right){\rm =}{{\rm x}}^{\beta },$ then the joint cdf \eqref{eq.FBEEW2} becomes
%\[{\rm F}\left({{\rm x}}_{{\rm 1}}{\rm ,\ }{{\rm x}}_{{\rm 2}}\right){\rm =}\left\{ \begin{array}{ll}
%{\left({\rm 1-}{{\rm e}}^{{\rm -}\lambda {{\rm x}}^{\beta }_{{\rm 1}}}\right)}^{{\alpha }_{{\rm 1}}{\rm +}{\alpha }_{{\rm 3}}}{\left({\rm 1-}{{\rm e}}^{{\rm -}\lambda {{\rm x}}^{\beta }_{{\rm 2}}}\right)}^{{\alpha }_{{\rm 2}}} & {\rm if\ \ \ }{{\rm x}}_{{\rm 1}}\le {{\rm x}}_{{\rm 2}} \\
%{\left({\rm 1-}{{\rm e}}^{{\rm -}\lambda {{\rm x}}^{\beta }_{{\rm 1}}}\right)}^{{\alpha }_{{\rm 1}}}{\left({\rm 1-}{{\rm e}}^{{\rm -}\lambda {{\rm x}}^{\beta }_{{\rm 2}}}\right)}^{{\alpha }_{{\rm 2}}{\rm +}{\alpha }_{{\rm 3}}} & {\rm if\ \ \  }{{\rm x}}_{{\rm 1}}{\rm >}{{\rm x}}_{{\rm 2}}. \end{array}
%\right.\]

If ${ H}\left({ x};\boldsymbol\xi \right){ =}{{ x}}^{\beta },$ then the joint cdf \eqref{eq.FBEEW2} becomes
\[{ F}\left({{ x}}_{{ 1}}{ ,\ }{{ x}}_{{ 2}}\right){ =}\left\{ \begin{array}{ll}
{\left({ 1-}{{ e}}^{{ -}\lambda {{ x}}^{\beta }_{{ 1}}}\right)}^{{\alpha }_{{ 1}}{ +}{\alpha }_{{ 3}}}{\left({ 1-}{{ e}}^{{ -}\lambda {{ x}}^{\beta }_{{ 2}}}\right)}^{{\alpha }_{{ 2}}} & { if\ \ \ }{{ x}}_{{ 1}}\le {{ x}}_{{ 2}} \\
{\left({ 1-}{{ e}}^{{ -}\lambda {{ x}}^{\beta }_{{ 1}}}\right)}^{{\alpha }_{{ 1}}}{\left({ 1-}{{ e}}^{{ -}\lambda {{ x}}^{\beta }_{{ 2}}}\right)}^{{\alpha }_{{ 2}}{ +}{\alpha }_{{ 3}}} & { if\ \ \  }{{ x}}_{{ 1}}{ >}{{ x}}_{{ 2}}. \end{array}
\right.\]
We call  this, bivariate exponentiated Weibull (BEW) distribution. By Theorem \ref{eq.FXi}, the marginal distributions of $X_1$ and $X_2$ are ${\rm EW}({\alpha }_1+{\alpha }_3,\lambda ,\beta )$ and ${\rm EW}\left({\alpha }_2+{\alpha }_3,\lambda ,\beta \right),$ respectively.
\subsection{ Bivariate generalized Gompertz distribution}

If $H\left(x;{\boldsymbol \xi }\right)={\beta }^{-1}(e^{\beta x}-1),$ then the joint cdf \eqref{eq.FBEEW2} becomes
\[F\left(x_1,\ x_2\right)=\left\{ \begin{array}{cc}
{\left(1-e^{-{\lambda \beta }^{-1}(e^{\beta x_1}-1)}\right)}^{{\alpha }_1+{\alpha }_3}{\left(1-e^{-{\lambda \beta }^{-1}(e^{\beta x_2}-1)}\right)}^{{\alpha }_2} & {\rm if}\ \ \ \ \ x_1\le x_2 \\
{\left(1-e^{-{\lambda \beta }^{-1}(e^{\beta x_1}-1)}\right)}^{{\alpha }_1}{\left(1-e^{-{\lambda \beta }^{-1}(e^{\beta x_2}-1)}\right)}^{{\alpha }_2+{\alpha }_3} & {\rm if}\ \ \ \ \ x_1>x_2, \end{array}
\right.\]
which is the joint cdf of bivariate generalized Gompertz (BGG) distribution introduced by
\cite{elSh-ib-be-13}.
%El-Sherpieny  et al. (2013).
By Theorem \ref{eq.FXi}, the marginal distributions of $X_1$ and $X_2$ are ${\rm GG}({\alpha }_1+{\alpha }_3,\lambda ,\beta )$ and ${\rm GG}\left({\alpha }_2+{\alpha }_3,\lambda ,\beta \right),$ respectively.

\subsection{ Bivariate exponentiated generalized Weibull-Gompertz distribution}

If $H\left(x;{\boldsymbol \xi }\right)=x^{\beta }(e^{\gamma x^{\delta }}-1),$ then the joint cdf \eqref{eq.FBEEW2} becomes
\[F\left(x_1,\ x_2\right)=\left\{ \begin{array}{ll}
{\left(1-e^{-\lambda x^{\beta }_1(e^{\gamma x^{\delta }_1}-1)}\right)}^{{\alpha }_1+{\alpha }_3}{\left(1-e^{-\lambda x^{\beta }_2(e^{\gamma x^{\delta }_2}-1)}\right)}^{{\alpha }_2} & {\rm if}\ \ \  x_1\le x_2 \\
{\left(1-e^{-\lambda x^{\beta }_1(e^{\gamma x^{\delta }_1}-1)}\right)}^{{\alpha }_1}{\left(1-e^{-\lambda x^{\beta }_2(e^{\gamma x^{\delta }_2}-1)}\right)}^{{\alpha }_2+{\alpha }_3} & {\rm if}\ \ \  x_1>x_2, \end{array}
\right.\]
which is the joint cdf of bivariate exponentiated generalized Weibull-Gompertz (BEGWG) distribution introduced by
\cite{elba-elda-mu-el-15}.
%EL-Damcese  et al. (2015).
By Theorem \ref{eq.FXi}, the marginal distributions of $X_1$ and $X_2$ are ${\rm EGWG}({\alpha }_1+{\alpha }_3,\lambda ,\beta ,\gamma ,\delta )$ and ${\rm EGWG}\left({\alpha }_2+{\alpha }_3,\lambda ,\beta ,\gamma ,\delta \right),$ respectively.

\subsection{ Bivariate exponentiated modified Weibull extension distribution}
If $H\left(x;{\boldsymbol \xi }\right)=\beta (e^{{(x/\beta )}^{\gamma }}-1),$ then the joint cdf \eqref{eq.FBEEW2} becomes
\[F\left(x_1,\ x_2\right)=\left\{ \begin{array}{ll}
{\left(1-e^{-\lambda \beta (e^{(x_1/\beta )^{\gamma }}-1)}\right)}^{{\alpha }_1+{\alpha }_3}{\left(1-e^{-\lambda \beta (e^{{(x_2/\beta )}^{\gamma }}-1)}\right)}^{{\alpha }_2} & {\rm if}\ \ \ \ \ x_1\le x_2 \\
{\left(1-e^{-\lambda \beta (e^{{(x_1/\beta )}^{\gamma }}-1)}\right)}^{{\alpha }_1}{\left(1-e^{-\lambda \beta ({??}^{{(x_2/\beta )}^{\gamma }}-1)}\right)}^{{\alpha }_2+{\alpha }_3} & {\rm if}\ \ \ \ \ x_1>x_2, \end{array}
\right.\]
which is the joint cdf of bivariate exponentiated modified Weibull extension (BEMWE) distribution introduced by
\cite{elgo-elmo-15}.
%El- Gohary and El- Morshedy (2015).
By Theorem \ref{eq.FXi}, the marginal distributions of $X_1$ and $X_2$ are ${\rm EMWE}({\alpha }_1+{\alpha }_3,\lambda ,\beta ,\gamma )$ and ${\rm EMWE}\left({\alpha }_2+{\alpha }_3,\lambda ,\beta ,\gamma \right),$ respectively.

\section{ Maximum likelihood estimation}
\label{sec.mle}
In this section, we first study the maximum likelihood estimations (MLE's) of the parameters. Then, we propose an Expectation-Maximization (EM) algorithm to estimate the parameters.

Let $\left(x_{11},x_{12}\right),\dots ,\left(x_{1n},x_{2n}\right)$ be an observed sample with size $n$ from BEEW distribution with parameters ${\boldsymbol \Theta }=\left({\alpha }_1,{\alpha }_2,{\alpha }_3,\lambda ,{\boldsymbol \zeta }\right)'$. Also, consider
\[I_0=\left\{i:x_{1i}=x_{2i}=x_i\right\},\ \ \ \ \ \ \ I_1=\left\{i:x_{1i}<x_{2i}\right\},\ \ \ \ \ \ I_2=\left\{i:x_{1i}>x_{2i}\right\},\ \ \ i=1,\dots ,n,\]
and
\[n_0=\left|I_0\right|,\ \ \ \ \ \ n_1=\left|I_1\right|,\ \ \ \ \ \ n_2=\left|I_2\right|,\ \ \ \ \ \ n=n_0+n_1+n_2.\]
Therefore, the log-likelihood function can be written as
\begin{eqnarray}\label{eq.lik}
\ell \left({\boldsymbol \Theta }\right)&=&\sum_{i\in I_1}{{\log  \left(f_1\left(x_{1i},x_{2i}\right)\right)\ }}+\sum_{i\in I_2}{{\log  \left(f_{{\rm 2}}\left(x_{1i},x_{2i}\right)\right) }}+\sum_{i\in I_0}{{\log  \left(f_0\left(x_i\right)\right) }}\nonumber\\
&=&
\left(2n_1+2n_2+n_0\right){\log  \left(\lambda \right)\ }+n_1{\log  \left({\alpha }_2\right) }+n_2{\log  \left({\alpha }_1\right) }+n_0{\log  \left({\alpha }_3\right) }\nonumber\\
&&+n_1{\log  \left({\alpha }_1+{\alpha }_3\right) }+n_2{\log  \left({\alpha }_2+{\alpha }_3\right)\ }+\sum_{i\in I_1\cup I_2}{{\log  \left(h\left(x_{1i};{\boldsymbol \xi }\right)\right) }}\nonumber\\
&&
+\sum_{i\in I_1\cup I_2}{{\log  \left(h\left(x_{2i};{\boldsymbol \xi }\right)\right) }}
+\sum_{i\in I_0}{{\log  \left(h\left(x_i;{\boldsymbol \xi }\right)\right) }}\nonumber\\
&&
+\left({\alpha }_1+{\alpha }_3-1\right)\left(\sum_{i\in I_1}{{\log  \left(1-e^{-\lambda H(x_{1i};{\boldsymbol \xi })}\right)\ }}+\sum_{i\in I_2}{{\log  \left(1-e^{-\lambda H(x_{2i};{\boldsymbol \xi })}\right) }}\right)\nonumber\\
&&
+\left({\alpha }_2-1\right)\sum_{i\in I_1}{{\log  \left(1-e^{-\lambda H\left(x_{2i};{\boldsymbol \xi }\right)}\right) }}
+\left({\alpha }_1-1\right)\sum_{i\in I_2}{{\log  \left(1-e^{-\lambda H\left(x_{1i};{\boldsymbol \xi }\right)}\right) }}\nonumber\\
&&
+\left({\alpha }_1+{\alpha }_2+{\alpha }_3-1\right)\sum_{i\in I_0}{{\log  \left(1-e^{-\lambda H\left(x_i;{\boldsymbol \xi }\right)}\right) }}\nonumber\\
&&
+\lambda \left(\sum_{i\in I_0}{x_i}+\sum_{i\in I_1\cup I_2}{x_{1i}}+\sum_{i\in I_1\cup I_2}{x_{2i}}\right),
\end{eqnarray}
where  $f_1$,  $f_2$ and $f_0$ are given in \eqref{eq.f1}, \eqref{eq.f2} and \eqref{eq.f0}, respectively.
 We can obtain the MLE's of the parameters by maximizing $\ell \left({\boldsymbol \Theta }\right)$ in \eqref{eq.lik} with respect to the unknown parameters. This is clearly a $(k+4)$-dimensional problem. However, no explicit expressions are available for the MLE's. We need to solve $(k+4)$ non-linear equations simultaneously, which may not be very simple. Therefore, we present an expectation-maximization (EM) algorithm to find the MLE's of parameters.
It may be noted that if instead of $(X_1,X_2)$, we observe $U_1$, $U_2$, and $U_3$, the MLE's of the parameters can be obtained by solving a two-dimensional optimization process, which is clearly much more convenient than solving a $(k+4)$-dimensional optimization process. For this reason, we treat this problem as a missing value problem.

Assumed that for the bivariate random vector$\ \left(X_1,X_2\right)$, there is an associated random vectors
\[{\Lambda }_1=\left\{ \begin{array}{ll}
0 & U_1>U_3 \\
1 & U_1<U_3 \end{array}
\right.\ \ \ \ \ \ \ \ {\rm and}\ \ \ \ \ \ \ {\Lambda }_2=\left\{ \begin{array}{cc}
0 & U_2>U_3 \\
1 & U_2<U_3. \end{array}
\right.\ \]

Note that if $X_1=X_2$, then${{\rm \ }\Lambda }_1={\Lambda }_2=0$. But if $X_1<X_2$ or $X_1>X_2$, then $({\Lambda }_1,{\Lambda }_2)$ is missing. If $\left(X_1,X_2\right)\in I_1$ then the possible values of $({\Lambda }_1,{\Lambda }_2)$ are $\left(1,0\right)$ or $(1,1)$, and If $\left(X_1,X_2\right)\in I_2$ then the possible values of $({\Lambda }_1,{\Lambda }_2)$ are $\left(0,1\right)$ or $(1,1)$ with non-zero probabilities.

Now, we are in a position to provide the EM algorithm. In the E-step of the EM-algorithm, we treat it as complete observation when they belong to $I_0$. If the observation belong to $I_1$, we form the `pseudo' log-likelihood function by fractioning $(x_1,x_2)$ to two partially complete ``pseudo'' observations of the form $(x_1,x_2,u_1\left({\boldsymbol \Theta }\right))$ and $(x_1,x_2,u_2\left({\boldsymbol \Theta }\right))$, where $u_1\left({\boldsymbol \Theta }\right)$ and $u_2\left({\boldsymbol \Theta }\right)$ are the conditional probabilities that $({\Lambda }_1,{\Lambda }_2)$ takes values $\left(1,0\right)$ and $(1,1)$, respectively. It is clear that
\[u_1\left({\boldsymbol \Theta }\right)=\frac{{\alpha }_1}{{\alpha }_1+{\alpha }_3},\ \ \ \ \ \ \ u_2\left({\boldsymbol \Theta }\right)=\frac{{\alpha }_3}{{\alpha }_1+{\alpha}_3}.\]

Similarly, If the observation belong to $I_2$, we form the `pseudo' log-likelihood function of the from $\left(y_1,y_2,v_1\left({\boldsymbol \Theta }\right)\right)$ and $\left(x_1,x_2,v_2\left({\boldsymbol \Theta }\right)\right)$, where $v_1\left({\boldsymbol \Theta }\right)$ and $v_2\left({\boldsymbol \Theta }\right)$ are the conditional probabilities that $({\Lambda }_1,{\Lambda }_2)$ takes values $\left(0,1\right)$ and $(1,1)$, respectively. Therefore,
\[v_1\left({\boldsymbol \Theta }\right)=\frac{{\alpha }_2}{{\alpha }_2+{\alpha }_3},\ \ \ \ \ \ \ v_2\left({\boldsymbol \Theta }\right)=\frac{{\alpha }_3}{{\alpha }_2+{\alpha }_3}.\ \]
For brevity, we write  $u_1\left({\boldsymbol \Theta }\right)$, $u_2\left({\boldsymbol \Theta }\right)$, $v_1\left({\boldsymbol \Theta }\right)$, $v_2\left({\boldsymbol \Theta }\right)$ as $u_1$, $u_2$, $v_1$, $v_2$, respectively.

\bigskip
\noindent \textbf{E-step:} Consider $b_i=E(N|y_{1i},y_{2i},{\boldsymbol \Theta })$. The log-likelihood function without the additive constant can be written as follows:
\begin{eqnarray*}
{\ell }_{{\rm pseudo}}\left({\boldsymbol \Theta }\right)&=& \left(n_0+2n_1+2n_2\right){\log  \left(\lambda \right) }+\left(u_1n_1+n_2\right){\log  \left({\alpha }_1\right)\ }+\left(n_1+v_1n_2\right){\log  \left({\alpha }_2\right) }\\
&&
+\left(n_0+u_2n_1+v_2n_2\right){\log  \left({\alpha }_3\right) }+\sum_{i\in I_0}{{\log  \left(h\left(x_i;{\boldsymbol \xi }\right)\right) }}+\sum_{i\in I_1\cup I_2}{{\log  \left(h\left(x_{1i};{\boldsymbol \xi }\right)\right) }}\\
&&
+\sum_{i\in I_1\cup I_2}{{\log  \left(h\left(x_{2i};{\boldsymbol \xi }\right)\right) }}+\left({\alpha }_1+{\alpha }_2+{\alpha }_3-1\right)\sum_{i\in I_0}{{\log  \left(1-e^{-\lambda H\left(x_i;{\boldsymbol \xi }\right)}\right) }}\\
&&
+\left({\alpha }_1+{\alpha }_3-1\right)\sum_{i\in I_1}{{\log  \left(1-e^{-\lambda H(x_{1i};{\boldsymbol \xi })}\right) }}\\
&&
+\left({\alpha }_2+{\alpha }_3-1\right)\sum_{i\in I_2}{{\log  \left(1-e^{-\lambda H(x_{2i};{\boldsymbol \xi })}\right)}}\\
&&
+\left({\alpha }_2-1\right)\sum_{i\in I_1}{{\log  \left(1-e^{-\lambda H(x_{2i};{\boldsymbol \xi })}\right) }}+\left({\alpha }_1-1\right)\sum_{i\in I_2}{{\log  \left(1-e^{-\lambda H(x_{1i};{\boldsymbol \xi })}\right) }}\\
&&
-\lambda \left(\sum_{i\in I_0}{H\left(x_i;{\boldsymbol \xi }\right)}+\sum_{i\in I_1\cup I_2}{H\left(x_{1i};{\boldsymbol \xi }\right)}+\sum_{i\in I_1\cup I_2}{H\left(x_{2i};{\boldsymbol \xi }\right)}\right)
\end{eqnarray*}

\bigskip
\noindent\textbf{M-step:} At this step, ${\ell }_{{\rm pseudo}}\left({\boldsymbol \Theta }\right)$ is maximized with respect to ${\alpha }_1,{\alpha }_2,{\alpha }_3,\lambda $ and ${\boldsymbol \xi }$. For fixed $\lambda $ and ${\boldsymbol \xi }$, the maximization occurs at
\begin{eqnarray}
{\hat{\alpha }}_1\left(\lambda ,{\boldsymbol \xi }\right)&=&\frac{u_1n_1+n_2}{\sum_{i\in I_0}{{W(x_i) }}+\sum_{i\in I_1\cup I_2}{{
W(x_{1i}) }}},\label{eq.a1hat}\\
{\hat{\alpha }}_2\left(\lambda \right)&=&\frac{n_1+v_1n_2}{\sum_{i\in I_0}{{W(x_i)} }+\sum_{i\in I_1\cup I_2}
{{W(x_{2i}) }}},\label{eq.a2hat}\\
\hat{\alpha }_3(\lambda )&=&\frac{n_0+u_2n_1+v_2n_2}{\sum_{i\in I_0}{{W(x_{i}) }}+\sum_{i\in I_1}{{W(x_{1i}) }}+\sum_{i\in I_2}{W(x_{2i}) }},\label{eq.a3hat}
\end{eqnarray}
where $W(x)=\log  \left(1-e^{-\lambda H\left(x;{\boldsymbol \xi }\right)}\right)$.
For fixed ${\alpha }_1,{\alpha }_2,{\alpha }_3$ and ${\boldsymbol \xi }$, ${\ell }_{{\rm pseudo}}\left({\boldsymbol \Theta }\right)$ is maximized with respect to $\lambda $ as a solution of the following equation:
\begin{equation}\label{eq.lamhat}
\frac{n_0+2n_1+2n_2}{{\rm g}(\lambda )}=\lambda,
\end{equation}
where
\begin{eqnarray*}
{\rm g}\left(\lambda \right)&=&-\left({\alpha }_1+{\alpha }_2+{\alpha }_3-1\right)\sum_{i\in I_0}{\frac{H\left(x_i;{\boldsymbol \xi }\right)e^{-\lambda H\left(x_i;{\boldsymbol \xi }\right)}}{1-e^{-\lambda H\left(x_i;{\boldsymbol \xi }\right)}}}\\
&&
-\left({\alpha }_1+{\alpha }_3-1\right)\sum_{i\in I_1}{\frac{H(x_{1i};{\boldsymbol \xi })e^{-\lambda H(x_{1i};{\boldsymbol \xi })}}{1-e^{-\lambda H(x_{1i};{\boldsymbol \xi })}}}
-\left({\alpha }_2+{\alpha }_3-1\right)\sum_{i\in I_2}{\frac{H(x_{2i};{\boldsymbol \xi })e^{-\lambda H(x_{2i};{\boldsymbol \xi })}}{1-e^{-\lambda H(x_{2i};{\boldsymbol \xi })}}}\\
&&
-\left({\alpha }_2-1\right)\sum_{i\in I_1}{\frac{H(x_{2i};{\boldsymbol \xi })e^{-\lambda H(x_{2i};{\boldsymbol \xi })}}{1-e^{-\lambda H(x_{2i};{\boldsymbol \xi })}}}
-\left({\alpha }_1-1\right)\sum_{i\in I_2}{\frac{H(x_{1i};{\boldsymbol \xi })e^{-\lambda H(x_{1i};{\boldsymbol \xi })}}{1-e^{-\lambda H(x_{1i};{\boldsymbol \xi })}}}\\
&&
+\sum_{i\in I_0}{H\left(x_i;{\boldsymbol \xi }\right)}
+\sum_{i\in I_1\cup I_2}{H\left(x_{1i};{\boldsymbol \xi }\right)}+\sum_{i\in I_1\cup I_2}{H\left(x_{2i};{\boldsymbol \xi }\right)}.
\end{eqnarray*}

Finally, for fixed ${\alpha }_1,{\alpha }_2,{\alpha }_3$ and $\lambda $, ${\ell }_{{\rm pseudo}}\left({\boldsymbol \Theta }\right)$ is maximized with respect to${\boldsymbol \ }{\boldsymbol \xi }$\textbf{ }as a solution of the following equation:
\begin{equation}\label{eq.xihat}
\frac{\partial }{\partial {\boldsymbol \xi }}\ {\ell }_{{\rm pseudo}}\left({\boldsymbol \Theta }\right)={\boldsymbol 0}.
\end{equation}

The following steps can be used to compute the MLE's of the parameters via the EM algorithm:

\noindent\textbf{Step 1}: Take some initial value of ${\boldsymbol \Theta }$\textbf{, }say ${{\boldsymbol \Theta }}^{{\rm (0)}}{\rm =}\left({\alpha }^{(0)}_1,{\alpha }^{(0)}_2,{\alpha }^{\left(0\right)}_3,{\lambda }^{\left(0\right)},{{\boldsymbol \xi }}^{(0)}\right)'$.

\noindent\textbf{Step 2}: Compute $u_1$, $u_2$, $v_1$, and $v_2$.

\noindent\textbf{Step 3: }Find $\hat{\lambda }$ by solving the equation \eqref{eq.lamhat}, say ${\hat{\lambda }}^{(1)}$.

\noindent\textbf{Step 4}: Find ${\boldsymbol \ }\hat{{\boldsymbol \xi }}$ by solving the equation \eqref{eq.xihat}, say ${\hat{{\boldsymbol \xi }}}^{(1)}$.

\noindent\textbf{Step 5}: Compute ${\hat{\alpha }}^{\left(1\right)}_i={\hat{\alpha }}_i({\hat{\lambda }}^{\left(1\right)},{\hat{{\boldsymbol \xi }}}^{(1)})$, $i=1,2,3$ from \eqref{eq.a1hat}-\eqref{eq.a3hat}.

\noindent\textbf{Step 6}: Replace ${{\boldsymbol \Theta }}^{{\rm (0)}}$ by ${\hat{{\boldsymbol \Theta }}}^{{\rm (1)}}{\rm =}\left({\hat{\alpha }}^{\left(1\right)}_1,{\hat{\alpha }}^{\left(1\right)}_2,{\hat{\alpha }}^{\left(1\right)}_3,{\hat{\lambda }}^{\left(1\right)},{\hat{{\boldsymbol \xi }}}^{\left(1\right)}\right)$, go back to step 1 and continue the process until convergence take place.

\section{Two real examples}
\label{sec.ex}

We consider BEEW distributions for fitting these two data sets. But, this family of distributions is a large class of distributions. Here, we consider six sub-models of BEEW distributions:  BGE, BGLFR, BEW, BGG, BEGWG, and BEMWE.  Some of them are suggested in literature.

Using the proposed EM algorithm, these models are fitted to the bivariate data set, and the MLE's and their corresponding log-likelihood values are calculated. The standard errors (s.e.) based on the observed information matrix are obtained.

For each fitted model, the Akaike Information Criterion (AIC), the corrected Akaike information criterion (AICC) and the Bayesian information criterion (BIC) are calculated. We also obtain the Kolmogorov-Smirnov (K-S) distances between the fitted distribution, the empirical distribution function, and the corresponding p-values (in parenthesis) for $X_1$, $X_2$ and ${\max  (X_1,X_2) }$.
Finally, we make use the likelihood ratio test (LRT) and the corresponding p-values for testing the BGE against other models.

\begin{example}

The data set is given from
\cite{meintanis-07}
%Meintanis (2007)
and is obtained from the group stage of the UEFA Champion's League for the years 2004-05 and 2005-2006. In addition,
\cite{ku-gu-09-BGE}
%Kundu and Gupta (2009-2010),
and
\cite{sa-ha-sm-ku-11}
%Sarhan et al. (2011)
analyzed this data. The data represent the football (soccer) data where at least one goal scored by the home team and at least one goal scored directly from a kick goal (like penalty kick, foul kick or any other direct kick) by any team have been considered. Here $X_1$ represents the time in minutes of the first kick goal scored by any team and $X_2$ represents the first goal of any type scored by the home team.

The results are given in Table \ref{tab.ex1}. It can be concluded that all six models are appropriate for this data set. But, the BGW and BGG distributions are better than other distributions.
\end{example}

\begin{table}

\caption{The MLE's, log-likelihood, AIC, AICC, BIC, K-S, and LRT statistics for six sub-models of BEEW distribution of first data set.}\label{tab.ex1}
\center
\begin{tabular}{|c|c|c|c|c|c|c|} \hline
 &   \multicolumn{6}{|c|}{Distribution} \\ \hline
Statistic & BGE & BGLFR & BEW & BGG & BEWG & BEMWE \\ \hline
${\hat{\alpha }}_1$ & 1.4452 & 0.4920 & 0.2179 & 0.6596 & 0.2474 & 0.1574 \\
(s.e.) & (0.4160) & (0.0810) & (0.6663) & (0.2559) & (0.1185) & (0.2276) \\ \hline
${\hat{\alpha }}_2$ & 0.4681 & 0.1661 & 0.0770 & 0.2366 & 0.0896 & 0.0573 \\
(s.e.) & (0.1879) & (0.0535) & (0.2219) & (0.1093) & (0.0498) & (0.0833) \\ \hline
${\hat{\alpha }}_3$ & 1.1704 & 0.4110 & 0.1880 & 0.5821 & 0.2223 & 0.1419 \\
(s.e.) & (0.2866) & (0.0331) & (0.3446) & (0.1964) & (0.1016) & (0.2009) \\ \hline
$\hat{\lambda }$ & 0.0390 & --- & 1.914e-4 & 0.0098 & 0.1622 & 0.0246 \\
(s.e.) & (0.0056) & --- & (1.83e-5) & (0.0061) & (0.6398) & (0.0526) \\ \hline
$\hat{\beta }$ & --- & 1.990e-4 & 3.7136 & 0.0304 & 0.4168 & 85.9181 \\
(s.e.) & --- & 1.237e-4 & (0.2811) & (0.0112) & (0.9648) & (34.1193) \\ \hline
$\hat{\gamma}$ & --- & 7.971e-4 & --- & --- & 2.624e-5 & 4.5054 \\
(s.e.) & --- & 1.497e-4 & --- & --- & 7.304e-5 & (2.0339) \\ \hline
$\hat{\delta }$ & --- & --- & --- & --- & 2.4645 & --- \\
(s.e.) & --- & --- & --- & --- & (0.5969) & --- \\ \hline
$-{\log  (\ell )\ }$ & 296.901 & 293.376 & 291.681 & 291.855 & 291.132 & 290.981 \\
AIC  & 601.801 & 596.752 & 593.361 & 593.710 & 596.263 & 593.962 \\
AICC  & 603.051 & 598.688 & 595.297 & 595.646 & 600.125 & 596.762 \\
BIC & 608.245 & 604.807 & 601.416 & 601.765 & 607.540 & 603.628 \\ \hline
K-S ($X_1$) & 0.1034 & 0.07082 & 0.0962 & 0.1042 & 0.1140 & 0.1182 \\
(p-value) & (0.8240) & (0.9925) & (0.8829) & (0.8157) & (0.7218) & (0.6789) \\ \hline
K-S ($X_2$) & 0.1001 & 0.0968 & 0.1167 & 0.1243 & 0.1196 & 0.1187 \\
(p-value) & (0.8527) & (0.8786) & (0.6939) & (0.6161) & (0.6644) & (0.6738) \\ \hline
K-S (${\max  (X_1,X_2)\ }$) & 0.1431 & 0.1104 & 0.0942 & 0.0984 & 0.1272 & 0.1366 \\
(p-value) & (0.4344) & (0.7574) & 0.8978 & (0.8661) & (0.5865) & (0.4940) \\ \hline
LRT & --- & 7.050 & 10.440 & 10.092 & 11.538 & 11.840 \\
(p-value) & --- & (0.0079) & (0.0012) & (0.0015) & (0.0091) & (0.0026) \\ \hline
\end{tabular}
\end{table}

\begin{example}

The data set was first published in  `Washington Post' and is available in
\cite{cs-we-89}.
%Csorgo and Welsh (1989).
It is represent the American Football League for the matches on three consecutive weekends in 1986. Here, $X_1$ represents the `game time' to the first points scored by kicking the ball between goal posts, and represents the `game time' to the first points scored by moving the ball into the end zone.
\cite{ku-gu-10modified}
%Kundu and Gupta (2010),
\cite{ja-ku-13weighted},
%Jamalizadeh and Kundu (2012),
and
\cite{ba-sh-14class}
%Balakrishna and Shiji (2014)
analyzed this data. We divided all the data by 100.
The results are given in Table \ref{tab.ex2}. It can be concluded that all six models are appropriate for this data set. But, the BGE distribution is better than other distributions.
\end{example}

\begin{table}

\caption{The MLE's, log-likelihood, AIC, AICC, BIC, K-S, and LRT statistics for six sub-models of BEEW distribution of second data set.}\label{tab.ex2}
\center

\begin{tabular}{|c|c|c|c|c|c|c|} \hline
 &  \multicolumn{6}{|c|}{Distribution} \\ \hline
Statistic & BGE & BGLFR & BEW & BGG & BEWG & BEMWE \\ \hline
${\hat{\alpha }}_1$ & 0.0921 & 0.0921 & 0.1367 & 0.0921 & 0.1501 & 0.1374 \\
(s.e.) & (0.0653) & (0.0667) & (0.1351) & (0.0653) & (0.2570) & (0.1355) \\ \hline
${\hat{\alpha }}_2$ & 0.5722 & 0.5722 & 0.8483 & 0.5722 & 0.9313 & 0.8523 \\
(s.e.) & (0.1614) & (0.1824) & (0.6283) & (0.1614) & (1.4720) & (0.6290) \\ \hline
${\hat{\alpha }}_3$ & 1.1519 & 1.1519 & 1.7113 & 1.1519 & 1.8788 & 1.7195 \\
(s.e.) & (0.2388) & (0.2945) & (1.2318) & (0.2388) & (2.9542) & (1.2328) \\ \hline
$\hat{\lambda }$ & 9.6187 & --- & 8.5587 & 9.6187 & 3.4632 & 3.0614 \\
(s.e.) & (1.5569) & --- & (1.9069) & (1.5590) & (2.8867) & (9.3275) \\ \hline
$\hat{\beta }$ & --- & 9.6187 & 0.8117 & 2.1e-12 & 0.5548 & 211.651 \\
(s.e.) & --- & (2.7455) & (0.2828) & (0.0455) & (0.0328) & (88.5725) \\ \hline
$\hat{\gamma }$ & --- & 2.351e-4 & --- & --- & 1.2553 & 0.8088 \\
(s.e.) & --- & 1.297e-4 & --- & --- & (0.8749) & (0.2814) \\ \hline
$\hat{\delta }$ & --- & --- & --- & --- & 0.1462 & --- \\
(s.e.) & --- & --- & --- & --- & (1.9357) & --- \\ \hline
${\log  (\ell )\ }$ & 36.670 & 36.670 & 36.857 & 36.670 & 36.859 & 36.857 \\
AIC & -65.340 & -63.340 & -63.714 & -63.340 & -59.717 & -61.714 \\
AICC & -64.258 & -61.673 & -62.048 & -61.673 & -56.423 & -59.314 \\
BIC & -58.389 & -54.652 & -55.026 & -54.651 & -47.553 & -51.288 \\ \hline
K-S ($X_1$) & 0.1808 & 0.1808 & 0.1678 & 0.1808 & 0.1679 & 0.1680 \\
(p-value) & (0.1282) & (0.1282) & (0.1872) & (0.1282) & (0.1869) & (0.1866) \\ \hline
K-S ($X_2$) & 0.1410 & 0.1411 & 0.1289 & 0.1410 & 0.1290 & 0.1291 \\
(p-value) & (0.3408) & (0.3408) & (0.4499) & (0.3408) & (0.4490) & (0.4484) \\ \hline
K-S (${\max  (X_1,X_2)\ }$) & 0.1350 & 0.1350 & 0.1197 & 0.1350 & 0.1198 & 0.1198 \\
(p-value) & (0.3929) & (0.3929) & (0.5438) & (0.3929) & (0.5428) & (0.5422) \\ \hline
LRT & --- & 0.000 & 0.374 & 0.000 & 0.378 & 0.374 \\
(p-value) & --- & 1.0000 & (0.5408) & 1.0000 & 0.9447 & 0.8294 \\ \hline
\end{tabular}

\end{table}

\section{Conclusions}
\label{sec.con}
In this paper we have introduced the bivariate exponentiated extended Weibull distribution whose marginals are exponentiated extended Weibull distributions. We discussed some statistical properties of the new bivariate model. Maximum likelihood estimates of the new class of distributions are discussed and we provided the observed Fisher information matrix. Two real data sets are used to show the usefulness of the new class.

\bibliographystyle{apa}
%\bibliography{D:/mypapers/myBIB}

\end{document}